%% file: arxiv.tex
\RequirePackage{amsmath}

\documentclass[a4paper,UKenglish,runningheads,envcountsame,draft]{llncs}

\usepackage{macros}
\usepackage[utf8]{inputenc}
\begin{document}

\title{Parameterised Complexity of\\ Model Checking and Satisfiability\\ in Propositional Dependence Logic\thanks{Funded by German Research Foundation (DFG), project ME 4279/1-2}}

\titlerunning{Parameterised MC and SAT in PDL}

\author{Yasir Mahmood 
\and Arne Meier 
}

\institute{Leibniz Universität Hannover, Institut für Theoretische Informatik
\email{\{mahmood,meier\}@thi.uni-hannover.de}}

\authorrunning{Y.~Mahmood and A.~Meier}

\maketitle             
\input{content}
\bibliography{main}
\newpage
\appendix
\input{appendix}

\end{document}

%% file: content.tex
\begin{abstract}
In this paper, we initiate a systematic study of the parameterised complexity in the field of Dependence Logics which finds its origin in the Dependence Logic of Väänänen from 2007.
We study a propositional variant of this logic (PDL) and investigate a variety of parameterisations with respect to the central decision problems.
The model checking problem (MC) of PDL is NP-complete (Ebbing and Lohmann, SOFSEM 2012).
The subject of this research is to identify a list of parameterisations (formula-size, formula-depth, treewidth, team-size, number of variables) under which MC becomes fixed-parameter tractable.
Furthermore, we show that the number of disjunctions or the arity of dependence atoms (dep-arity) as a parameter both yield a paraNP-completeness result.
Then, we consider the satisfiability problem (SAT) which classically is known to be NP-complete as well (Lohmann and Vollmer, Studia Logica 2013).
There we are presenting a different picture: under team-size, or dep-arity SAT is paraNP-complete whereas under all other mentioned parameters the problem is in FPT.
Finally, we introduce a variant of the satisfiability problem, asking for teams of a given size, and show for this problem an almost complete picture.
\keywords{Propositional Dependence Logic \and Parameterised Complexity \and Model Checking \and Satisfiability}

\end{abstract}
\section{Introduction}

The logics of dependence and independence are a recent innovation studying such central formalisms occurring in several areas of research: computer science, logic, statistics, game theory, linguistics, philosophy, biology, physics, and social choice theory \cite{grdel_et_al:DR:2016:5508}. 
Jouko Väänänen~\cite{DBLP:books/daglib/0030191}
 initiated this subfield of research in 2007, and nowadays, it is a vibrant area of study~\cite{DBLP:books/daglib/0037838}.
Its focus widened from initially first-order dependence logic further to modal logic~\cite{v08}, temporal logics~\cite{DBLP:conf/mfcs/KrebsMV018,DBLP:conf/time/KrebsMV15}, probabilistic logics~\cite{DBLP:conf/foiks/0001HKMV18}, logics for independence~\cite{DBLP:journals/logcom/KontinenMSV17}, inclusion logics~\cite{DBLP:journals/apal/Galliani12,hkmv19}, multi-team semantics~\cite{DBLP:journals/amai/DurandHKMV18}, and poly-team semantics~\cite{DBLP:conf/lfcs/HannulaKV18}.

In this paper, we study a sub-logic of the modal variant which is called propositional dependence logic (\PDL) \cite{DBLP:journals/apal/YangV16,DBLP:journals/tocl/HannulaKVV18}. 
The main concept also in this logic, the \emph{dependence atom} $\depa{P}{Q}$, intuitively states that the variables $p\in P$ functionally determine the values of the variables $q\in Q$.
As functional dependence only makes sense on sets of assignments, which Väänänen called \emph{teams}, team-semantics are the heart of the satisfaction relation $\models$ in this logic.
Formally, a team $T$ is a set of classical propositional assignments $t\colon \VAR \to\setdefinition{0,1}$, and $T\models\depa{P}{Q}$ if and only if for all $t,t'\in T$, we have that $t$ and $t'$ agree on the variables in $P$ implies $t$ and $t'$ agree on variables in $Q$.

The model checking question ($\MC$), given a team $T$ and a $\PDL$-formula $\varphi$, asks if $T\models\varphi$ is true.
The satisfiability problem ($\SAT$), given a $\PDL$-formula $\varphi$, asks for the existence of a team $T$ such that $T\models\varphi$.
It is known that $\MC$ as well as $\SAT$ are $\NP$-complete by Ebbing and Lohmann~\cite{DBLP:conf/sofsem/EbbingL12}, respectively, by Lohmann and Vollmer~\cite{DBLP:journals/sLogica/LohmannV13}.
These authors classify the complexity landscape of even operator-fragments of \PDL yielding a deep understanding of these problems from a classical complexity point of view.
For an overview of how other atoms (e.g., inclusion, or independence) influence the complexity of these problems consider the tables in the work Hella et~al.~\cite{hkmv19}. 

\begin{example}
We illustrate an example from relational databases providing understanding of team logics.
\begin{table}[t]
	\centering
	\begin{tabular}{c@{\; }c@{\; }c@{\; }c}\toprule
\texttt{Instructor} 	& 	\texttt{Room}		& 	\texttt{Time} 	& 	\texttt{Course} \\ \toprule
Antti		&	A.10		&	09.00	&	Logic 	\\
Antti		&	A.10		&	11.00	&	Statistics \\
Antti		&	B.20		&	15.00	&	Algebra \\
Jonni		&	C.30		&	10.00	&	LAB	 \\
Juha		&	C.30		&	10.00	&	LAB 	\\
Juha		&	A.10		&	13.00	&	Statistics	\\
\bottomrule
	\end{tabular}\qquad
	\begin{tabular}{c@{\; }c@{\; }c@{\; }c}\toprule
$i_1i_2$ 	& 	$r_1r_2$		& 	$t_1t_2t_3$ 	& 	$c_1c_2$ \\ \toprule
00		&	11 		&	110	&	11 	\\
00		&	11		&	111	&	00 \\
00		&	00		&	000	&	01 \\
01		&	01		&	001	&	10	 \\
10		&	01		&	001	&	10 	\\
10		&	11		&	010	&	00	\\
\bottomrule
	\end{tabular}	
	\caption{(Left) An example database with $4$ attributes and universe size $15$.\newline 
	(Right) An encoding with $\lceil\log_2(3)\rceil+\lceil\log_2(3)\rceil+\lceil\log_2(5)\rceil+\lceil\log_2(4)\rceil$ many propositional variables.}\label{database:example}
\end{table}	
Table~\ref{database:example} depicts a database which can be expressed in $\PDL$ via binary encoding of the possible entries for the attributes.
The set of rows then corresponds to a team $T$. 
The database satisfies two functional dependencies: 
\[
\depa{\{\texttt{Room}, 	\texttt{Time}\}}{\{\texttt{Course}\}}\text{ and }\depa{\{\texttt{Instructor}, 	\texttt{Time}\}}{\{\texttt{Room}, \texttt{Course}\}}.
\]
Whereas, it does not satisfy $\depa{\{\texttt{Room}, \texttt{Time}\}}{\{\texttt{Instructor}\}}$ as witnessed by the tuples (Juha, C.30, 10, LAB) and (Jonni, C.30, 10, LAB). 
Formally, we have that 
\[
	T\models \depa{\{\texttt{Room}, 	\texttt{Time}\}}{\{\texttt{Course}\}}\land\depa{\{\texttt{Instructor}, \texttt{Time}\}}{\{\texttt{Room}, \texttt{Course}\}},
\]
but
\[
	T\not\models\depa{\{\texttt{Room}, \texttt{Time}\}}{\{\texttt{Instructor}\}}.
\]
Notice that in propositional logic, we cannot express a table of so many values.
As a result, we need to binary encode the values of each column separately.
This might cause a logarithmic blow-up (by binary encoding the universe values for each column) in the parameter values, for example, it influences the number of variables.
Furthermore, one also has to rewrite variables in the occurring formulas accordingly.
For instance, as in Table~\ref{database:example}, for $\depa{\{\texttt{Room}, \texttt{Time}\}}{\{\texttt{Instructor}\}}$ this would yield the formula $\depa{\{r_1, r_2,t_1,t_2,t_3\}}{\{i_1,i_2\}}$.
The parameters discussed in this paper correspond to the already encoded values. 
This means that there is no need in considering this blow-up as in this example.
\end{example}

Often, when a problem is shown to be intrinsic hard, a possible way to further unravel the true reasons for the intractability is the framework of parameterised complexity theory~\cite{DBLP:series/txcs/DowneyF13}.
Here, one aims for a more fine-grained complexity analysis involving the study of parameterisations and how they pin causes for intractability substantially.
One distinguishes two runtimes of a different quality: $f(k)\cdot p(|x|)$ versus $p(|x|)^{f(k)}$, where $f$ is an arbitrary computable function, $p$ is a polynomial, $|x|$ the input length and $k$ the value of the parameter.
Clearly, both runtimes are polynomial in $x$ for each fixed $k$ but the first one is much better as the polynomial degree is independent of the parameter's value.
Problems that can be solved with algorithms running in a time of the first kind are said to be fixed-parameter tractable (or $\FPT)$.
Whereas, problems of category two are in the complexity class $\complClFont{XP}$.
It is known that $\FPT\subsetneq\XP$~\cite{DBLP:series/txtcs/FlumG06}.
Whenever runtimes of the form $f(k)\cdot p(|x|)$ are considered with respect to nondeterministic machines, one studies the complexity class $\para\NP\supseteq\FPT$.
In between these two classes a presumably infinite $\complClFont{W}$-hierarchy is contained: $\FPT\subseteq\W1\subseteq\W2\subseteq\cdots\subseteq\para\NP$.
It is unknown whether any of these inclusions is strict. 
Showing $\W1$-hardness of a problem intuitively corresponds to being intractable in the parameterised world.

The area of research of parameterised problems is tremendously growing and often provides new insights into the inherent difficulty of the studied problems \cite{DBLP:books/sp/CyganFKLMPPS15}.
However, the area of dependence logic is rather blank with respect to this direction of research, only Meier and Reinbold~\cite{DBLP:conf/foiks/MeierR18} investigated the (parameterised) enumeration complexity of a fragment of $\PDL$ recently.
As a subject of this research, we want to initiate and to further push a study of the parameterised complexity of problems in these logics.
\paragraph{Applications.} 
The teams in the team semantic bear a close resemblance with the relations studied in relational database theory. 
Moreover, dependence atoms are analogous to functional dependencies in the context of database systems. 
The $\MC$ problem for dependence logic, for example, is equivalent to determining whether a relation in the database satisfies a functional dependency. 
The teams of $\pdl$ also relate to the information states of inquisitive logic \cite{DBLP:conf/tbillc/CiardelliGR11}; a semantic framework for the study of the notion of meaning and information exchange among agents.

\begin{table}[t]
	\centering
	\begin{tabular}{ccccc} \toprule
	Parameter 	& $\MC$ 	  &  $\SAT$     & $\mSAT$   \\\midrule 
	\formulatw 	 
	& 	$\para\NP^{\ref{thm: formulatw}}$
	&  $\FPT^{\ref{sat:formulatw}}$ 
	& ? \\ 
	\formulateamtw 	 
	& $\FPT^{\ref{mc-formulateamtw}}  $	
	& see above
	& see above \\ 
	\teamsize  	 	
	& $\FPT^{\ref{mc:teamsize}}$	
	& $\para\NP^{\ref{thm:sat-teamsize-arity}}$	 
	& $\para\NP^{\ref{m-sat:arity}}$\\ 
	\formula 	 
	& $\FPT^{\ref{mc:remaining-cases}}$	 
	&  $\FPT^{\ref{sat:remaining-cases}}$ 	
	& $\FPT^{\ref{m-sat:all-cases}}$		\\
	\formuladepth 	 	
	& $\FPT^{\ref{mc:remaining-cases}}$	
	&  $\FPT^{\ref{sat:remaining-cases}}$
	& $\FPT^{\ref{m-sat:all-cases}}$		\\ 		
	\variables   	& $\FPT^{\ref{mc:remaining-cases}}$ 
	& $\FPT^{\ref{sat:remaining-cases}}$ 	
	& $\FPT^{\ref{m-sat:all-cases}}$		\\ 
	\splits  	 	
	& $\para\NP^{\ref{mc:splits}}$
	& $\FPT^{\ref{sat:splits}}$
	& ?		\\ 
	\arity
	& $\para\NP^{\ref{mc:arity}}$
	& $\para\NP^{\ref{thm:sat-teamsize-arity}}$
	& $\para\NP^{\ref{m-sat:arity}}$\\
	\bottomrule\\
	\end{tabular}
	\caption{Complexity classification overview showing the results of the paper with pointers to the theorems. All results are completeness results.  The question mark symbol means that the precise complexity is unknown.}\label{tbl:overview} 
\end{table}

\paragraph{Contributions.}
We study a wealth of parameters, also relevant from the perspective of database theory.
Specifically, the parameter \teamsize corresponds to the number of entries in the database and \variables is the number of attributes.
The parameter \formulatw denotes how much interleaving is present among the attributes in the query and \arity bounds the size of functional dependencies in the query.
Furthermore, the parameter $\formulateamtw$ bounds the interleaving between a query and the database, $\formula$ limits the size of the query, $\formuladepth$ restricts the nesting depth of the query, and $\splits$ controls the unions in relational algebra queries.
With respect to all parameters, we study $\MC$ and $\SAT$.
Furthermore, we introduce a satisfiability variant $\mSAT$, which has an additional unary input $m\in\N$, and asks for a satisfying team of size exactly $m$.

In Table~\ref{tbl:overview}, we give an overview of our results. 
In this article, we prove dichotomies for $\MC$ and $\SAT$: depending on the parameter the problem is either fixed-parameter tractable or $\para\NP$-complete.
Only the satisfiability variant under the parameters $\formulatw$ and $\splits$ resist a complete classification and are left for further research.

\paragraph*{Related work.}
The notion of treewidth is due to Robertson and Seymour~\cite{DBLP:journals/jct/RobertsonS86}. 
The study of the complexity of bounded treewidth query evaluation is a vibrant area of research \cite{DBLP:conf/stoc/GroheSS01,DBLP:journals/toct/ChenM15,DBLP:journals/mst/0001M15,DBLP:conf/icdt/ChenM15,DBLP:conf/pods/ChenM16,DBLP:conf/lics/ChenM17}. 
As stated earlier, the formulas of dependence logic correspond to the functional dependencies in the database context.
Bläsius~et~al.~\cite{blsius_et_al:LIPIcs:2017:6920} study the parameterised complexity of dependency detection. 
The problem is defined as, given a database $T$ and a positive integer $k$ whether there is a non-trivial functional dependency of size ($\arity$ in our notion) at most $k$ that is satisfied by $T$. 
These authors prove that this problem is $\W2$-complete.

\paragraph*{Organisation of the article.}
At first, we introduce some required notions and definitions in (parameterised) complexity theory, dependence logic, and propositional logic. 
Then we study the parameterised complexity of the model checking problem.
We proceed with the satisfiability problem and study a variant of it.
Finally, we conclude and discuss open questions.
For results marked with a ($\star$) their proof can be found in the appendix. 

\section{Preliminaries}
In this paper, we assume familiarity with standard notions in complexity theory~\cite{DBLP:books/daglib/0092426} such as the classes $\NP$ and $\Ptime$.

\subsection{Parameterised Complexity}
We will recapitulate some relevant notion of parameterised complexity theory, now. 
For a broader introduction consider the textbook of Downey and Fellows~\cite{DBLP:series/txcs/DowneyF13}, or Flum and Grohe~\cite{DBLP:series/txtcs/FlumG06}.
A parameterised problem (PP) $\Pi\subseteq\Sigma^*\times \mathbb{N} $ consists of tuples $(x,k)$, where $x$ is called the \emph{instance} and $k$ the \emph{(value of the) parameter}.
\begin{definition}[Fixed-parameter tractable and $\para\NP$]
	Let $\Pi$ be a PP over $\Sigma^*\times\mathbb{N}$.
	We say that $\Pi$ is \emph{fixed-parameter tractable} (or is in the class $\FPT$) if there exists a deterministic algorithm $\mathcal{A}$ deciding $\Pi$ in time $f(k)\cdot |x|^{O(1)}$ for every input $(x,k)\in\Sigma^*$, where $f$ is a computable function.
	If $\mathcal{A}$ is a nondeterministic algorithm instead, then $\Pi$ belongs to the class $\para\NP$.
\end{definition}
Let $P$ be a PP over $\Sigma^*\times\N$.
Then the \emph{$\ell$-slice of $P$}, for $\ell\ge0$, is the set $P_\ell\dfn\{\,x\mid (x,\ell)\in P\,\}$. 
It is customary to use the notation $O^\star(f(k))$ to denote the runtime dependence only on the parameter and to ignore the polynomial factor in the input. 
We will use the following result from parameterised complexity theory to prove \para\NP-hardness results. 
\begin{proposition}[{\cite[Theorem~2.14]{DBLP:series/txtcs/FlumG06}}]\label{slice-NP-result}
	Let $P$ be a PP.
	If there exists an $\ell\ge0$ such that $P_\ell$ is $\NP$-complete, then $P$ is $\para\NP\complete$.
\end{proposition} 
\noindent Moreover, we will use the following folklore result to get several upper bounds.
\begin{proposition}\label{parameter-bound}
	Let $Q$ be a problem such that $(Q, k)$ is in $\FPT$ and let $\ell$ be another parameter such that $k\leq f(\ell)$ for some computable function $f$, then $(Q, \ell)$ is also in \FPT. 
\end{proposition}

\subsection{Propositional Dependence Logic}
Let $\VAR$ be a countably infinite set of variables.
The syntax of propositional dependence logic (\PDL) is defined via the following EBNF:
$$
\varphi\ddfn \top \mid 
			 \bot \mid
			 x \mid
			 \lnot x\mid
			 \varphi\lor\varphi\mid
			 \varphi\land\varphi\mid
			 \depa{X}{Y}\mid
			 \lnot \depa{X}{Y},
$$
where $\top$ is \emph{verum}, $\bot$ is \emph{falsum}, $x\in\VAR$ is a variable, $X,Y\subset\VAR$ are finite sets of variables, $\depa{\cdot}{\cdot}$ is called the \emph{dependence atom}, and the disjunction $\lor$ is also called \emph{split-junction}.
Observe that we only consider atomic negation.
We let $\pl$ be defined as the $\PDL$-formulas without $\depa{\cdot}{\cdot}$.
Finally, the set $X$ in $\depa{X}{Y}$ can be empty, giving rise to formulas of the form $\depa{}{Y}$.
To simplify the notation, when either set in the arguments of $\depa{X}{Y}$ is singleton then we write, for example, $\depa{x}{y}$ instead of $\depa{\{x\}}{\{y\}}$.
\begin{definition}[Team semantics]
	Let $\varphi,\psi$ be $\PDL$-formulas and $P,Q\subseteq\VAR$ be two finite sets of variables.
	A \emph{team} $T$ is a set of assignments $t\colon\VAR\to\setdefinition{0,1}$.
	Furthermore, we define the satisfaction relation $\models$ as follows, where $T\models \top$ is always true, $T\models \bot$ is never true, and $T\models \neg \depa{P}{Q}$ iff $T=\emptyset$:
	\begin{alignat*}{3}
		&T\models x && \text{ iff } \quad&& \forall t\in T: t(x)=1 \\
		&T\models \lnot x && \text{ iff } && \forall t\in T: t(x)=0 \\
		&T\models \varphi\land\psi && \text{ iff } && T\models\varphi\text{ and }T\models\psi \\
		&T\models \varphi\lor\psi && \text{ iff } && \exists T_1 \exists T_2 (T_1\cup T_2=T): T_1\models\varphi\text{ and }T_2\models\psi \\
		&T\models \depa{P}{Q}\quad && \text{ iff } && \forall t,t'\in T: \bigwedge_{p\in P}t(p)=t'(p)\text{ implies }\bigwedge_{q\in Q}t(q)=t'(q)
	\end{alignat*}
\end{definition}

Observe that for formulas of the form $\depa{}{Q}$ the team has to be constant with respect to $Q$.
That is why such atoms are called \emph{constancy atoms}. 
Note that in literature there exist two semantics for the split-junction operator: \emph{lax} and \emph{strict} semantics (e.g., Hella et~al.~\cite{hkmv19}).
Strict semantics requires the ``splitting of the team'' to be a partition whereas lax semantics allow an ``overlapping'' of the team.
We use lax semantics here.
Notice that the computational complexity for $\SAT$ and $\MC$ in $\PDL$ are the same irrespective of the considered semantics.
Furthermore, our proofs work for both semantics. 
Also further note that allowing an unrestricted negation operator dramatically increases the complexity of $\SAT$ in this logic to $\complClFont{ATIME\text-ALT(exp, poly)}$ (alternating exponential time with polynomially many alternations) as shown by Hannula~et~al.~\cite{DBLP:journals/tocl/HannulaKVV18}.
That is one reason why we stick to atomic negation.

In the following, we define three well-known formula properties which are relevant to results in the paper.
A formula $\phi$ is \emph{flat} if, given any team $T$, we have that $T\models \phi \iff \setdefinition{s}\models \phi$ for every $s\in T$.
A logic $\mathcal L$ is \emph{downwards closed} if for every $\mathcal L$-formula $\phi$ and team $T$, if $T\models \phi$ then for every $P\subseteq T$ we have that $P\models \phi$.
A formula $\phi$ is \emph{2-coherent} if for every team $T$, we have that $T\models \phi \iff \setdefinition{s_i,s_j}\models \phi$ for every $s_i,s_j\in T$.
The classical \pl-formulas are flat. 
This also implies that for \pl-formulas, the truth value is evaluated under each assignment individually, consequently, the semantics is the usual Tarski semantic.
Moreover, \pdl is downwards closed and every dependence atom is 2-coherent. 

\subsection{Representation of Inputs as Graphs}
As we will consider specific structural parameters, we need to agree on a representation of formulas, respectively, teams.
Classically, propositional formulas were represented via different kinds of graphs (e.g., Gaifman graph, primal graph) \cite{DBLP:series/faia/SamerS09}.
However, in this setting usually CNF-formulas are considered. 
Coping with this restriction, Lück~et~al.\ \cite{DBLP:journals/tocl/LuckMS17} defined syntax circuits for temporal logic formulas that also allow arbitrary formulas.
In our setting, we continue this direction and define the syntax (or formula) circuit with respect to a $\pdl$-formula.

An important observation regarding the graph representation for the $\pdl$-formulas is due to Grädel \cite{DBLP:journals/tcs/Gradel13}.
In the usual setting for logics with team semantics, we take the syntax tree and not the associated circuit, that is, we distinguish between different occurrences of the same subformula. 
The reason for this choice is that a formula $\phi\lor \phi$ is not equivalent to $\phi$, and in its evaluation, different teams are entitled to the two occurrences of $\phi$ in the formula.
Consequently, the well-formed formulas of $\pdl$ can be seen as binary trees with leaves as atomic subformulas (variables and dependence atoms).
\begin{example}
The team $\{00,01,10,11\}$ satisfies $\depa{x}{y}\lor \depa{x}{y} $, even though it does not satisfy $\depa{x}{y}$.	
\end{example}
Notice that according to the graph representation of formulas as trees, the treewidth (Def. \ref{treewidth}) of a \pdl-formula is already $1$. 
As a consequence, the $1$-slice of each problem is \NP-hard and both problems ($\MC$ and $\SAT$) are \para\NP-complete when parameterised by the treewidth of the syntax tree of a \pdl-formula.
For this reason we consider the syntax circuit rather than the syntax tree as a graph structure.

Given an instance $\langle T, \Phi \rangle$ of the model checking problem, where $\Phi$ is a $\pdl$-formula with propositional variables $\setdefinition{x_1, \ldots, x_n}\subseteq\VAR$ and $T=\setdefinition{s_1, \ldots s_m}$ is a team of assignments $s_i\colon\VAR\to\setdefinition{0,1}$. 
Then we consider the graph-structure $\strucA_{T,\Phi}$ with vocabulary $\tau_\problemind$ and represent the formula by its syntax circuit.
Henceforth, we write $\strucA$ instead of $\strucA_{T,\Phi}$ when it is clear that our input instance is ${T,\Phi}$ and define the vocabulary as
\begin{multline*}
\tau_{\problemind} \dfn \{\, \SF{}^1, \succcurlyeq^2, r, \var{}^1, \NEG^2, \CONJ^3, \DISJ^3, \DEP^3,\\ \inteam{}^1, \istrue{}^2,\isfalse{}^2, c_1, \ldots,c_m\,\},
\end{multline*}
where the superscripts denote the arity of each relation.
The set of vertices $\univA $ of the graph is $\SubForm{\Phi} \cup \setdefinition{c_1^\strucA,\ldots, c_m^\strucA}$, where $\SubForm{\Phi}$ denotes the set of subformulas of $\Phi$. 
\begin{itemize}
	\item \SF{} and \var{} are unary relations representing `is a subformula of $\Phi$' and `is a variable in $\Phi$' respectively.
	\item $\succcurlyeq$ is a binary relation such that $\phi \succcurlyeq^\strucA \psi$ iff $\psi$ is the immediate subformula of $\phi$ and $r$ is a constant symbol representing $\Phi$.
	\item The set $\setdefinition{c_1,\ldots,c_m}$ encodes the team $T$, where each $c_i$ is interpreted as $ c_i^{\strucA} \in {\strucA}$ and each $c_i$ corresponds to an assignment $s_i \in T$ for $i \leq m$. 
	\item $\inteam{c}$ is true if and only if $c\in \setdefinition{c_1,\ldots,c_m}$.
	\item $\istrue{}$ and $\isfalse{}$ relate variables with the team elements. $\istrue{c, x}$ (resp., $\isfalse{c, x}$) is true if and only if $x$ is mapped $1$ (resp., $0$) by the assignment interpreted by $c$. 
\end{itemize}

The remaining relations interpret how subformulas are related to each other.
\begin{definition}[Gaifman graph]
	Given a team $T$ and a $\pdl$-formula $\Phi$, the \emph{Gaifman graph} $G_{T,\Phi}=(\univA,E)$ of the $\tau_{T,\Phi}$-structure $\mathcal A$ is defined as 
	$$E\dfn\setdefinition{\setdefinition{u,v}\mid u,v\in\univA, \text{$u$ and $v$ share a tuple in a relation in }\calA}.$$
\end{definition}

%

\begin{definition}[Treewidth]\label{treewidth}
The \emph{tree-decomposition} of a given graph $G=(V,E)$ is a tree $T=(B,E_T)$, where the vertex set $B\subseteq\mathcal P(V)$ is called \emph{bags} and $E_T$ is the edge relation such that the following is true: 
(i) $\bigcup_{b\in B}=V$ 
(ii) for every $\setdefinition{u,v}\in E$ there is a bag $b\in B$ with $u,v\in b$, and 
(iii) for all $v\in V$ the restriction of $T$ to $v$ (the subtree with all bags containing $v$) is connected.
The \emph{width} of a given tree-decomposition $T=(B,E_T)$ is the size of the largest bag minus one: $\max_{b\in B}|b|-1$.
The \emph{treewidth} of a given graph $G$ is the minimum over all widths of tree-decompositions of $G$.
\end{definition}
Observe that if $G$ is a tree then the treewidth of $G$ is one.
Intuitively, one can say that treewidth accordingly is a measure of tree-likeliness of a given graph.
\begin{example}
	Figure~\ref{fig:example-tw} represents a graph (in middle) with a tree-decomposition (on the right).
	Since the largest bag is of size $3$, the graph has a treewidth of $2$.
%
%
%
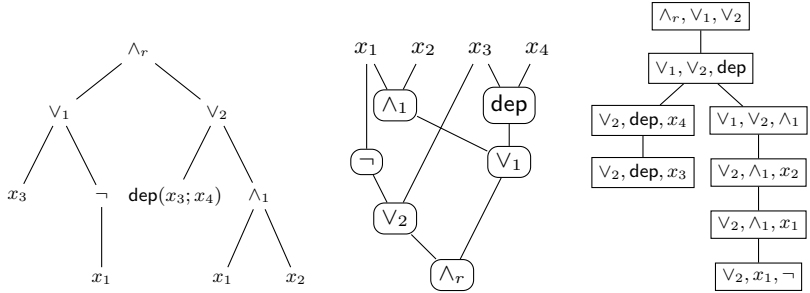
\begin{figure}[t]
\centering
\begin{tikzpicture}[scale=0.85, 
	level 1/.style={sibling distance=7.5em, level distance= 3 em},
	level 2/.style={sibling distance=4em, level distance= 4 em}, 
	level 3/.style={sibling distance=3.5em, level distance= 4 em}, 
	every node/.style={scale=0.8},
	edge from parent/.style={thin,-,black, draw},
	 leaf/.style = {draw, circle}]
\node  {\textcolor{black}{$\wedge_r$}}
		child {node {$\lor_1$}
			child {node {$x_3$}}
			child {node {$\neg$}
				child {node {$x_1$}}}}	
		child {node {$\lor_2$}	
				child{node {$\depa{x_3}{x_4}$ }}
				child{node {$\land_1$}
					child {node {$x_1$}}
					child {node {$x_2$}}}
				};
\end{tikzpicture} 
\quad
\begin{tikzpicture}[gate/.style={inner sep=1mm,draw,rounded corners,rectangle},scale=.75]
		\node (x1) at (0,0) {$x_1$};
		\node (x2) at (1,0) {$x_2$};
		\node (x3) at (2,0) {$x_3$};
		\node (x4) at (3,0) {$x_4$};
		
		\node[gate] (and1) at (.5,-1) {$\land_1$};
		\node[gate] (dep) at (2.5,-1) {$\mathsf{dep}$};

		\node[gate] (not) at (0,-2) {$\lnot$};
		\node[gate] (or1) at (2.5,-2) {$\lor_1$};

		\node[gate] (or2) at (.5,-3) {$\lor_2$};
		
		\node[gate] (andr) at (1.5,-4) {$\land_r$};
		
		\foreach \f/\t in {x1/and1,x2/and1,x3/dep,x4/dep,and1/or1,dep/or1,x1/not,not/or2,x3/or2,or2/andr,or1/andr}{
			\draw[-] (\f) -- (\t);
		}
	\end{tikzpicture}
\quad
\begin{tikzpicture}[scale=0.85, level distance= 2.5 em, sibling distance= 5.5 em,
	every node/.style={draw, scale=0.8},
	edge from parent/.style={thin,-,black, draw},
	 leaf/.style = {draw, circle}]
\node  {\textcolor{black}{$\wedge_r,\lor_1,\lor_2$}}
	child {node {$\lor_1,\lor_2,\mathsf{dep}$} 
		child {node {$\lor_2,\mathsf{dep},x_4$}
			child {node {$\lor_2,\mathsf{dep},x_3$}
				}}		
		child {node {$\lor_1,\lor_2,\land_1$}	
				child{node {$\lor_2,\land_1,x_2$}
					child{node {$\lor_2,\land_1,x_1 $}
						child{node {$\lor_2,x_1,\neg $}}}}}};
\end{tikzpicture}

\caption{An example syntax tree (left) with the corresponding circuit graph (middle) and a tree decomposition (right) for $(x_3\lor \neg x_1) \land [(\depa{x_3}{x_4}) \lor (x_1\land x_2)]$.}\label{fig:example-tw}
\end{figure}

\end{example}
\subsection{Considered Parameterisations.}
We consider eight different parameters for all three problems ($\MC$, $\SAT$ and $\mSAT$). 
These include $\formulatw$, $\formulateamtw$, $\teamsize$, $\formula$,\linebreak $\variables$, $\formuladepth$, $\splits$ and $\arity$.
All these parameters arise naturally in the problems we study. 
$\splits$ denotes the number of times a split junction $(\lor)$ appears in a formula and $\variables$ denotes the total number of propositional variables.
The parameter $\formuladepth$ is the depth of the syntax tree of $\Phi$, that is, the length of the longest path from root to any leaf in the tree. 
Arity of a dependence atom $\depa{P}{Q} $ is the length of $P$ and $\arity$ is the maximum arity of a dependence atom in the input formula.

Regarding treewidth, notice first that for the $\MC$ problem, we can also include the assignment-variable relation in the graph representation.
This yields two treewidth notions $\formulatw$ and $\formulateamtw$, the name emphasises whether the team is also part of the graph.
$\formulatw$ is the treewidth of the syntax circuit of input formula alone whereas $\formulateamtw$ comes from the syntax circuit when team elements are also part of the representation .
Clearly, $\formulateamtw$ is only relevant for the $\MC$ problem.

The following lemma proves relationships between several of the aforementioned parameters.
The notation $\kappa(\problemind)$ stands for the parameter $\kappa$ of the input instance $(\problemind)$.
\begin{figure}[t]
\centering
	\begin{tikzpicture}[auto, sibling distance= 2cm, level distance = 1.5 cm, scale=.5,  node/.style = {}, node1/.style = {}, edge/.style = {},x=1.25cm ] 

   \node[node] at (-4,2) (team) {$\teamsize$};
   \node[node] at (0,4) (var) {$\variables$};
   \node[node] at (4,4) (form) {$\formula$};
   \node[node] at (4,0) (depth) {$\formuladepth$};
   \node[node] at (0,0) (tw) {$\formulateamtw$}; 

 	\draw[stealth'-] (team) to node {${L\ref{para:relations}} $} (var) ; 
 	\draw[stealth'-] (var) to node {${L\ref{para:relations}} $} (form) ; 
 	\draw[dashed, stealth'-] (team) to node {${L\ref{para:tw-teamsize}} $} (tw) ; 
 	\draw[dashed,stealth'-] (var) to node {${L\ref{para:tw-teamsize}}$} (tw) ; 
 	\draw[stealth'-] (form) to node {${L\ref{para:relations}} $} (depth) ; 
\end{tikzpicture}
\caption{The relationship among different parameters. The direction of arrow in $p \leftarrow q$ implies that bounding $q$ results in bounding $p$. The dotted line indicates that the parameter bounds either (minimum) of the given two. $Li$ means the Lemma $i$.}\label{fig:parameters}
\end{figure}
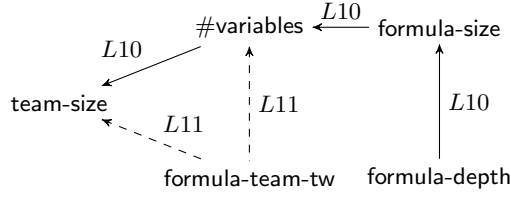 
\begin{lemma}\label{para:relations}
Given a team $T$ and a formula $\Phi$ then 
\begin{samepage}
\begin{enumerate}
	\item $\teamsize(\problemind) \leq 2^{\variables(\problemind)}$
	\item $\teamsize(\problemind ) \leq 2^{\formula(\problemind)}$
	\item $\formula(\problemind) \leq 2^{2\cdot \formuladepth(\problemind) }$
\end{enumerate}
\end{samepage}
\end{lemma}
\begin{proof}
	If a \pdl-formula $\Phi$ has $m$ variables then there are $2^m$ many assignments and the maximum size for a team is $2^m$. 
	As a result, we have $\teamsize \leq 2^{\variables}$.
	Furthermore, the number of variables in a \pdl-formula $\Phi$ is bounded by the $\formula$ and as a consequence, we have $2^{\variables} \leq 2^{\formula}$.
	This proves the second claim.

	If a formula $\Phi$ has $\formuladepth=d$ then there are $\leq 2^d$ leaves in the (binary) syntax tree of $\Phi$ and $\leq 2^d$ internal nodes.
	Then $\formula \leq 2^{2d}$ is true. 
\end{proof}
Now we prove the following non-trivial lemma stating that treewidth of the structure $\calA_\problemind$ bounds either the team size or the number of variables. 
This implies that bounding the treewidth of the structure also bounds either of the two parameters. 
Recall that we talk about the $\formulateamtw$ of the Gaifman graph underlying the structure $\calA_\problemind$ that encodes the $\MC$ question.
\begin{lemma}\label{para:tw-teamsize}
	Let $\langle T,\Phi\rangle$ be a given $\MC$ instance. 
	Then the following relationship between parameters is true,
	$$\formulateamtw(\problemind) \geq \min \setdefinition{\teamsize(\problemind), \variables(\problemind)}$$
\end{lemma}
\begin{proof}
	We prove that if $\formulateamtw$ is smaller than the two then such a decomposition must have cycles and hence cannot be a tree  decomposition.
    The proof idea uses the fact that in the Gaifman graph $\calA_\problemind$, every team element is related to each variable. 
    As a consequence, in any tree decomposition, the assignment-variable relations `\istrue{}' and `\isfalse{}' cause some bag to have their size larger than either the team size or the number of variables (based on which of the two values is smaller).
    We consider individual bags corresponding to an edge in the Gaifman graph due to the relations from $\tau_\problemind$.
    Let $\setdefinition{x_1,\ldots,x_n}$ be the variables that also appear as leaves in the formula tree $\langle \SubForm{\Phi}, \succcurlyeq, \Phi \rangle$. 
	
	Consider a minimal tree decomposition $\langle B_T,\prec \rangle$ for the Gaifman graph of $\calA$.
	Denote by $B_{x_i,c_j}$ the bag that covers the edge between a variable $x_i$ and an assignment-element $c_j$, that is, either $\istrue{x_i,c_j}$ or $\isfalse{x_i,c_j}$ is true. 
	Moreover, denote by $B_{x_i, \alpha_r}$ the bag covering the edges between a variable $x_i$ and its immediate $\succcurlyeq$-predecessor $\alpha_r$. 
	Recall that in the formula part of the Gaifman graph, there is a path from each variable $x_i$ to the formula $\Phi$ due to $\succcurlyeq$. 
	This implies that there exists a minimal path between any pair of variables in the Gaifman graph, and this path passes through some subformula $\Psi$ of $\Phi$.
	Let $B_{x,\alpha_1}, B_{\alpha_1,\alpha_2}, \ldots$ $B_{\alpha_{n},\Psi}, \ldots B_{\Psi, \beta_n},\ldots B_{\beta_2,\beta_1},B_{\beta_1,y}$ be the sequence of bags that covers $\succcurlyeq^\calA$-edges between $x$ and $y$.
	Without loss of generality, we assume that all these bags are distinct.
	Now, for any pair  $x,y$ of variables, the bags $B_{x,c_i}$ and $B_{y,c_i}$ contain $c_i$ for each $i\leq m$ and as a consequence, we have either of the following two cases.
	\begin{itemize}
		\item The two bags are equal, that is $B_{x, c_i}=B_{y, c_i}$ and as a consequence, we have $|B_{x, c_i}| \geq 3$ because $B_{x,c_i}$ contains at least $x,y,c_i$.
		Moreover, if this is true (otherwise case two applies) for each pair of variables, then there is a single bag $B_{c_i}$ that contains all the variables and the element $c_i$. 
		This means the maximum bag size must be larger than the total number of variables, a contradiction.
		\item Every path between $B_{x,c_i}$ and $B_{y,c_i}$ contains $c_i$. 
%
		We know that if a $B_{x,\alpha_1}$-$B_{\beta_1,y}$-path between $x$ and $y$ due to $\succcurlyeq$ exist, then the bags $B_{x, c_i}$ and $B_{y, c_i}$ cannot be incident because this will produce a cycle, a contradiction again.
%
%
		Now, for two different assignment-elements $c_i, c_j$, consider the bags $B_{y, c_i},B_{y, c_j}$.
		If these two bags are incident then $B_{x, c_j},B_{y, c_j}$ cannot be incident and the path between $B_{x, c_j},B_{y, c_j}$ must contain $c_j$. 
		Notice that both $B_{y, c_i},B_{y, c_j}$ and $B_{x, c_j},B_{y, c_j}$ cannot be incident since this would, again, create a cycle.
		Consequently, the only possible case is that either $B_{y, c_i}$ and $B_{y, c_j}$ are not incident and every path between these bags contains $y$, or $B_{x, c_j}$ and $B_{y, c_j}$ are not incident and every path between these bags contains $c_j$.
		Also see Figure~\ref{fig:tree-decomposition} explaining this situation.
		\end{itemize}
Finally, since this is true for all the variables and all the elements $c_i$ with $i\leq m$ this proves that either there is a bag that contains all the variables, or there is one that contains all $c_i$'s.
The remaining case that there are cycles in the tree decomposition is not applicable. 
This proves the claim and completes the proof to the lemma.
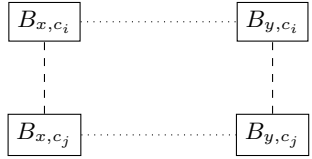
\begin{figure}[t]
\centering

\begin{tikzpicture}[auto, sibling distance= 2cm, level distance = 1.5 cm, scale=.75,  node/.style = {draw, shape=rectangle}, node1/.style = {}, edge/.style = {} ] 

   \node[node] at (-4,2) (31) {$B_{x,c_j}$};
    \node[node] at (0,2) (32) {$B_{y,c_j}$};

   \node[node] at (-4,4) (41) {$B_{x,c_i}$};
    \node[node] at (0,4) (42) {$B_{y,c_i}$};

 	\draw[dotted] (31) to node {$$} (32) ; 
 	\draw[dotted] (41) to node {$$} (42) ; 
 	\draw[dashed] (31) to node {$$} (41) ; 
 	\draw[dashed] (32) to node {$$} (42) ; 
\end{tikzpicture}
\caption{The rectangles represent bags corresponding to a variable-assignment relation. 
If the $c_i$-bags do not contain $c_j$-nodes, then
there can be only either dotted or dashed edges between the bags to avoid cycles.}\label{fig:tree-decomposition}
\end{figure}
\end{proof}
\noindent The following corollary is immediate due to previous lemma.
\begin{corollary}\label{cor:tw-teamsize}
	Let $\Phi\in\PDL$ and $T$ be a team. 
	Then $\formulateamtw(\problemind)$ bounds $\teamsize(\problemind)$. 
\end{corollary} 
\begin{proof}
	If $\formulateamtw \geq \variables$ then bounding $\formulateamtw$ bounds $\variables$ which in turn bounds $\teamsize$ because $\teamsize \leq 2^{\variables}$.
	Otherwise we already have $\formulateamtw \geq \teamsize$ according to Lemma~\ref{para:tw-teamsize}.
\end{proof}
\section{Parameterised Complexity of Model Checking in PDL}
In this section, we study the $\MC$ question under various parameterisations. 
Table~\ref{tbl:overview} contains a complete list of the results. 

\begin{proposition}[{\cite[Thm. 3.2]{DBLP:conf/sofsem/EbbingL12}}]\label{prop:mc-np-complete}
	$\MC$ is $\NP$-complete.
\end{proposition}

\begin{theorem}\label{thm: formulatw}
	$\MC$ parameterised by $\formulatw$ is $\para\NP$-complete.
\end{theorem}
\begin{proof}
	The upper bounds follows from Proposition~\ref{prop:mc-np-complete}.
	For lower bound we prove that $1$-slice of the problem is \NP-hard by reducing from 3SAT.
	The reduction provided by Ebbing and Lohmann (Prop.~\ref{prop:mc-np-complete}) uses Kripke semantics (as they aim for a modal logic related result).
	We slightly modify it to fit our presentation (the correctness proof is the same).
	Let $\Phi \dfn C_1\land \ldots C_m$ be a 3CNF over the variables $\{x_1,\ldots, x_n\}$.
	We form an instance $\langle T,\Psi \rangle$ of $\pdl\text-\MC$ such that $\VAR(\Psi)=\{p_1,\ldots,p_n, r_1\ldots,r_n\}$.
	The team $T= \{s_1,\ldots, s_m\}$ contains $m$ assignments where each assignment $s_i\colon\VAR(\Psi)\rightarrow \{0,1\} $ is defined as follows,
	\begin{alignat*}{2}
		& s_i(p_j) =  s_i(r_j) =  1  && \text{ if } x_j \in C_i,\\
		& s_i(p_j) = 0,  s_i(r_j) =  1  &&\text{ if } \neg x_j \in C_i,\\
		& s_i(p_j) =  s_i(r_j) =  0  &&\text{ otherwise. } 
	\end{alignat*}
Finally, let $\Psi \dfn \bigvee\limits\limits_{j=1}^{n}\left(r_j\land \depa{}{p_j}\right)$.
The proof of $T\models \Psi \iff \Phi \in \SAT$ follows from \cite[Thm. 3.2]{DBLP:conf/sofsem/EbbingL12}.
Notice that none other parameter except \formulatw is fixed in advance.
The syntax circuit of $\Psi$ yields a tree, as a consequence, $\formulatw = 1$.
This completes the proof.  
\end{proof}
Notice that the formula in the reduction from $\problemFont{3SAT}$ has fixed arity for any dependence atom (that is, $\arity =0$).
As a consequence, we obtain the following corollary.
\begin{corollary}\label{mc:arity}
		$\MC$ parameterised by \arity is $\para\NP$-complete.
\end{corollary} 

The main source of difficulty in the model checking problem seems to be the split-junction operator. 
For a team of size $k$ and a formula with only one split-junction there are $2^k$ many candidates for the correct split and each can be verified in polynomial time. 
As a result, an exponential runtime in the input length seems necessary. 
However, if the team size ($k$) is considered as a parameter then the problem can be solved in polynomial time with respect to the input size and exponentially in the parameter. 
We consider both parameters (\teamsize and \splits) in turn.
\begin{theorem}\label{mc:teamsize}
	$\MC$ parameterised by $\teamsize$ is in \FPT.
\end{theorem}
\begin{proof}
We claim that Algorithm~{\ref{alg:mc-teamsize}} solves the task in fpt-time.
The correctness follows from the fact that the procedure is simply a recursive definition of truth evaluation of \pdl-formulas in bottom-up fashion.
%

\begin{algorithm}[t]
\caption{check($T,\Phi$), recursive bottom-up algorithm solving the $\MC$ parameterised by \teamsize.}\label{alg:mc-teamsize}
\SetKwInOut{Input}{Input}
\SetKwInOut{Output}{Output}
\LinesNumbered
\DontPrintSemicolon

\Input{A $\PDL$-formula $\Phi$ and a team $T$}
\Output{true if $T\models \Phi$, otherwise false}

\lForEach{non-root node $v$ in the syntax tree}{
	$L_v=\setdefinition{\emptyset}$
}

\ForEach(// find all possible sub-teams for $\ell$){leaf $\ell$ of the syntax tree}{
	$L_\ell = \setdefinition{\emptyset}$\;
	\ForEach{$P\subseteq T$}{
		\lIf{$\ell = X$ and $\forall s\in P: s(X)=1$}{ 
			$L_\ell \leftarrow L_\ell \cup \setdefinition{P}$
		}
		\lElseIf{$\ell = \neg X $ and $\forall s\in P: s(X)=0 $}{ 
			$L_\ell \leftarrow L_\ell \cup \setdefinition{P}$
		}
		\ElseIf{$\ell=\depa{Q}{r}$ and $\forall s_i\forall s_j \smash{\underset{q \in Q}{\bigwedge}} s_i(q)= s_j (q) \Rightarrow s_i(r)= s_j (r) $}{
			$L_\ell \leftarrow L_\ell \cup \setdefinition{P}$
		}
	}
}
\ForEach{$\alpha_1, \alpha_2$ with $\alpha = \alpha_1\circ \alpha_2$ and $L_{\alpha_i}\neq\emptyset$ for $i=1,2$}{
	\ForEach{$P\in L_{\alpha_1}$, $Q\in L_{\alpha_2}$}{ 
		\lIf{$\circ=\land$ and $P=Q$}{
			$L_\alpha \leftarrow L_\alpha \cup \setdefinition{P}$
		}
		\lElseIf{$\circ=\lor$}{
			$L_\alpha \leftarrow L_\alpha \cup \setdefinition{P\cup Q}$	
		}
	}
}
	
	\leIf{$T \in L_\Phi $}{\Return true}{\Return false}
\end{algorithm}
Recall that the input formula $\Phi$ is a binary tree.
The procedure starts at the leaf level by checking whether for each subformula $\alpha$ and each subteam $P\subseteq T$, $P \models \alpha$.
Then recursively, if $P \models \alpha_i$ for $i=1,2$ and there is a subformula $\alpha$ such that $\alpha = \alpha_1\land \alpha_2$ then it answers that $P \models \alpha$.
Moreover, if $P_i \models \alpha_i$ for $i=1,2$ and there is a subformula $\alpha$ such that $\alpha = \alpha_1\lor \alpha_2$ then it answers that $P \models \alpha$ where $P= P_1\cup P_2$.

The first loop runs in $O^\star(2^k)$ steps for each leaf node and there are $|\Phi|$ many iterations, which gives a running time of $|\Phi|\cdot O^\star(2^k)$, where $\teamsize = k$.
At each inner node, there are at most $2^k$ candidates for $P$ and $Q$ and
as a consequence, at most $2^{2k}$ pairs that need to be checked.
This implies that the loop for each inner node can be implemented in $O^\star(2^{2k})$ steps. 
Furthermore, the loop runs once for each pair of subformulas $\alpha_1, \alpha_2$ such that $\alpha_1\circ \alpha_2$ is also a subformula of $\Phi$.
This gives a running time of $|\Phi|\cdot O^\star(2^{2k})$ for this step.
Finally, in the last step a set of size $k$ needs to be checked against a collection containing $2^k$ such sets, this can be done in $k\cdot O(2^k)$ steps.

We conclude that the above procedure solves the $\MC$ problem in $p(|\Phi|)\cdot O(2^{2k})$ steps for some polynomial $p$.
The fact that we do not get a blow-up in the number of subformulas is due to the reason that the formula tree is binary.
The procedure operates on a pair of subformulas in each step and the label size ($|L_\alpha|$) at the end of this step is again bounded by $2^k$. 
\end{proof}

Regarding the parameter $\splits$, we show $\para\NP$-completeness by reducing from the $3$-colouring problem ($\problemFont{3COL}$) and applying Proposition~\ref{slice-NP-result}.
The idea of the reduction from $\problemFont{3COL}$ is to construct a team as shown in Figure~\ref{fig:example-team} in combination with the disjunction of three times the formula $\bigwedge_{e_k= \setdefinition{v_i, v_j}}  \depa{y_k}{x_i}.$ 
Intuitively, vertices of the graph correspond to assignments in the team and the three splits then map to the three colours.
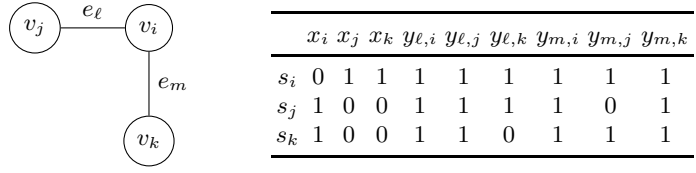
\begin{figure}[t]
\centering
\begin{tikzpicture}[auto, sibling distance= 2cm, level distance = 1.5 cm, scale=.75,  node/.style = {draw, shape =circle}, edge/.style = {} ] 
    \node[node] at (-2,0) (1) {$v_j$};
    \node[node] at (0,0) (2) {$v_i$};
    \node[node] at (0,-2) (3) {$v_k$};
 	\draw (1) to node {$e_\ell$} (2) ; 
 	\draw  (2) to node  {$e_m$}(3);
 	
 	\node[anchor=west] at (2,-1) {\begin{tabular}{cccccccccc}\toprule 
		 & $x_i$ &$x_j$ &$x_k$ &$y_{\ell,i}$ &$y_{\ell,j}$ &$y_{\ell,k}$&$y_{m,i}$ &$y_{m,j}$ &$y_{m,k}$ \\\midrule
		 $s_i$ & 0 &1 &1 &1 & 1 &1 &1  &1 & 1 \\
		$s_j$ & 1 &0 &0 &1 & 1 &1 &1  &0 & 1\\
		$s_k$ & 1 &0 &0 &1 & 1 &0 &1  &1 & 1\\
\bottomrule
		\end{tabular}};
\end{tikzpicture}
\caption{A graph $\calG: \langle\setdefinition{v_i,v_j,v_k},\setdefinition{e_l,e_m} \rangle$ and a corresponding team.}\label{fig:example-team}
\end{figure}
\begin{theorem}[$\star$] \label{mc:splits}
	$\MC$ parameterised by \splits is $\para\NP$-complete.
\end{theorem}
The following cases then can be easily deduced.
\begin{theorem}[$\star$]\label{mc:remaining-cases}
	$\MC$ under the parameterisations \formula, \formuladepth or \variables is  $\FPT$.
\end{theorem}
Finally, the case for $\formulateamtw$ follows due to Corollary~\ref{cor:tw-teamsize} in conjunction with the $\FPT$ result for $\teamsize$ (Lemma~\ref{mc:teamsize}).
\begin{corollary}\label{mc-formulateamtw}
	$\MC$ parameterised by $\formulateamtw$ is $\FPT$.
\end{corollary}

\section{Satisfiability}
In this section, we study $\SAT$ under various parameterisations, so the question of whether there exists a team $T$ for a given formula $\Phi$ such that $T\models\Phi$. 
Notice first, that the question is equivalent to finding a singleton team.
This is since $\pdl$ is downwards closed.
Consequently, if there is a satisfying team, then a singleton team satisfies the formula.
As a result, team semantics coincides with the usual Tarskian semantics. 
This facilitates, for example, determining the truth value of disjunctions in the classical way.
Accordingly, simplifying the notation a bit, for $\SAT$ we now look for an assignment rather than a singleton team that satisfies the formula. 
\begin{corollary}\label{lem:sat-formulateamtw}
	The problem $\SAT$ under the parameterisations $\formulateamtw$ and $\formulatw$ is same. 
\end{corollary}
\noindent The following result is obtained by classical $\SAT$ being $\NP$-complete \cite{DBLP:conf/stoc/Cook71,levin73}.
\begin{corollary}\label{thm:sat-teamsize-arity}
	$\SAT$ under the parameterisations  $\teamsize$, or $\arity$ is \linebreak $\para\NP$-complete.
\end{corollary}
\begin{proof}
The $1$-slice regarding \teamsize (a singleton team is the same as an assignment), and the $0$-slice regarding \arity (no dependence atoms at all) is \NP-hard.	
\end{proof}


Turning towards treewidth, notice that classical propositional $\SAT$ is fixed-parameter tractable when parameterised by treewidth due to Samer and Szeider~\cite[Thm. 1]{DBLP:journals/jda/SamerS10}.
However, we are unable to immediately utilise their result because Samer and Szeider study CNF-formulas and we have arbitrary formulas instead.
Yet, Lück~et~al.~\cite[Cor. 4.7]{DBLP:journals/tocl/LuckMS17} studying temporal logics under the parameterised approach, classified, as a byproduct, the propositional satisfiability problem with respect to arbitrary formulas to be fixed-parameter tractable.
\begin{corollary}\label{sat:formulatw}
	$\SAT$ parameterised by $\formulatw$ is in \FPT.
\end{corollary}
\begin{proof}
As before, we need to find a singleton team.
This implies that split-junctions have the same semantics as classical disjunctions and dependence atoms are always satisfied.
So replacing every occurrence of a dependence atom $\depa{P}{Q}$ by $\top$ yields a propositional logic formula.
This substitution does not increase the treewidth.
Then the result follows by Lück~et~al.~\cite[Cor. 4.7]{DBLP:journals/tocl/LuckMS17}.
\end{proof}

Now, we turn towards the parameter $\splits$.
We present a procedure that constructs a satisfying assignment $s$ such that $s\models \Phi $ if there is one and  otherwise it answers no. 
The idea is that this procedure needs to remember the positions where a modification in the assignment is possible.
We show that the number of these positions is bounded by the parameter $\splits$.

Consider the syntax tree of $\Phi$ where, as before, multiple occurrences of subformulas are allowed. 
The procedure starts at the leaf level with satisfying assignment candidates (partial assignments, to be precise). 
Reaching the root it confirms whether it is possible to have a combined assignment or not. 
We assume that the leaves of the tree consist of literals or dependence atoms. 
Accordingly, the internal nodes of the tree are only conjunction and disjunction nodes. 
The procedure sets all the dependence atoms to be trivially \emph{true} (as we satisfy them via every singleton team). 
Additionally, it sets all the literals satisfied by their respective assignment. 
Ascending the tree, it checks the relative conditions for conjunction and disjunction by joining the assignments and thereby giving rise to conflicts. 
A conflict arises (only at a conjunction node) when two assignments are joined with contradicting values for some variable. 
At this point, it sets this variable $x$ to a conflict state $c$. 
At disjunction nodes the assignment stores that it has two options and keeps the assignments separately. 

Joining a \emph{true}-value from a dependence atom affects the assignment only at disjunction nodes. 
This corresponds to the intuition that a formula of the form $\depa{P}{Q}\lor \psi$ is true under any assignment. 
Whereas, at a conjunction node, when an assignment $s$ joins with a \emph{true}, the procedure returns the assignment $s$. 
Since at a split the procedure returns both assignments, for $k$ splits there could be $\leq 2^k$-many assignment choices. 
At the root node if at least one assignment is consistent then we have a satisfying assignment. 
Otherwise, if all the choices contain conflicts over some variables then there is no such satisfying singleton team. 

\begin{theorem}[$\star$]\label{sat:splits}
	$\SAT$ parameterised by \splits is $\FPT$. 
	Moreover, there is an algorithm that solves the problem in $O(2^{\splits(\Phi)}\cdot |\Phi|^{O(1)})$ for any $\Phi\in\PDL$.
\end{theorem}
\noindent The remaining cases follow easily.
\begin{theorem}[$\star$]\label{sat:remaining-cases}
	$\SAT$ parameterised by \variables,  \formula \linebreak or \formuladepth is $\FPT$.
\end{theorem} 

\subsection*{A satisfiability variant.}
The shown results suggest that it might be interesting to study the following variant of $\SAT$, in which we impose an additional input $1^m$ (unary encoding) with $m\ge2$ and ask for a satisfying team of size $m$. 
Let us call the problem $\mSAT$. 
We wish to emphasise that $\mSAT$ is not the same as the $\SAT$ parameterised by \teamsize. 

\begin{theorem}[$\star$]\label{m-sat:all-cases}
$\mSAT$ under the parameterisations $\variables$, $\formula$, or $\formuladepth$ is $\FPT$.
\end{theorem}

Neither the arity of the dependence atoms nor the team-size alone are fruitful parameters which follows from Corollary~\ref{thm:sat-teamsize-arity}.
\begin{corollary}\label{m-sat:arity}
	$\mSAT$ parameterised by $\teamsize$, or $\arity$ is $\para\NP$-complete.
\end{corollary} 

\section{Conclusion}
In this paper, we started a systematic study of the parameterised complexity of model checking and satisfiability in propositional dependence logic.
For both problems, we exhibited a complexity dichotomy (see Table~\ref{tbl:overview}): depending on the parameter, the problem is either $\FPT$ or $\para\NP$-complete.
Interestingly, there exist parameters for which $\MC$ is easy, but $\SAT$ is hard (\teamsize) and \emph{vice versa} (\splits).

In the end, we introduced a satisfiability question which also asks for a team of a given size ($\mSAT$).
This has not been studied at all in the setting of team logics, yet.
We pose it as an interesting problem to study.
Here, we leave the cases for \splits and \formulatw open for further research.

As future work, we want to study combinations of the studied parameters, e.g., $\splits+\arity$.
This parameter is quite interesting, as \arity alone is always hard for all three problems, whereas adding \splits allows for $\SAT$ to reach $\FPT$.
It is also interesting to observe that in both of our reductions for proving hardness of $\MC$ under the parametrisation $\splits$ and $\arity$, if $\arity$ is fixed then $\splits$ is unbounded and vice versa. 

Another important question for future research is to consider the parameterised version of validity and implication problem for $\pdl$. 
Finally, we aim, besides answering the open cases, to study further operators such as independence and inclusion atoms.

%% file: appendix.tex
\section{Appendix}

\begin{restatetheorem}[mc:splits]
\begin{theorem}
	$\MC$ parameterised by \splits is $\para\NP$-complete.
\end{theorem}	
\end{restatetheorem}
\begin{proof}
We show a reduction from the question of whether a given graph is $3$-colourable.
Given an instance $\langle \calG\rangle$ where $\calG = (V,E)$ is a graph. 
We map this input to an instance $\langle (T, \Phi), 2 \rangle$ where $T$ is a team, and $\Phi$ is a $\pdl$-formula with $2$ split-junctions. 
Let $V= \setdefinition{v_1,\ldots, v_n}$ be the vertex set and $E=\setdefinition{e_1,\ldots e_m}$ the given set of edges. 
Then we define
 $$\VAR(\Phi) =\setdefinition{x_1,\ldots, x_n}\cup \setdefinition{y_{1,1}, \ldots, y_{1,n}, \ldots, y_{m,1}, \ldots, y_{m,n} }.$$ 
That is, we have (1) a variable $x_i$ corresponding to each node $v_i$ and (2) a variable $y_{j,k}$ corresponding to each edge $e_j$ and each node $v_k$. 
For convenience, we will sometimes write ${y_j} $ instead of $(y_{j,1}\ldots y_{j,n}) $ when it is clear that we are talking about the tuple of variables corresponding to the edge $e_j$. 
Consequently, we have an $n$-tuple of variables $y_{j}$ for each edge $e_j$, where $1\leq j\leq m$. 
The idea of the team that we construct is that there is an assignment $s_i$ corresponding to each node $v_i$ that encodes the neighbourhood of $v_i$ in the graph. 
The assignment $s_i$ also encodes all the edges that $v_i$ participates in. 
This is achieved by mapping each variable $y_{\ell,j}$ in tuple $y_\ell$ to $1$ if $v_j\in e_\ell $ whereas $y_{\ell,j}= 0$ if $v_j\not \in e_\ell $ and for every $j\neq i$, $y_{\ell,j}= 1$.  
Figure~\ref{fig:example-team} illustrates an example to get an intuition on this construction. 

Formally, we define the team as follows.
\begin{enumerate}
	\item If $\mathcal G$ has an edge $e_\ell = \setdefinition{v_i,v_j}$ then we set $s_i(x_j) = 1$ and $s_j(x_i) = 1$, and let $s_i(y_{\ell,1})=\ldots =s_i(y_{\ell,n})=1$ as well as $s_j(y_{\ell,1})=\ldots= s_j(y_{\ell,n})=1$ 
	\item For the case $v_j\not \in e_{\ell}$, we set $s_j(y_{\ell,j}) = 0$ and for the remaining indices $s_j(y_{\ell,i}) = 1$.
	\item Since, we can assume w.l.o.g.\ the graph has no loops (self-edges) we always have $s_i(x_i)=0$ for all $1\leq i\leq n$.
\end{enumerate}

Consequently, two assignments $s_i, s_j$ agree on $y_k$ if the corresponding edge $e_k$ is the edge between $v_i$ and $v_j$, and we have $s_i(y_k) = 1 = s_j(y_k)$.\\
Now let $\Phi$ be the following formula
$$ \Phi \dfn \bigwedge\limits_{e_k= \setdefinition{v_i, v_j}}  \depa{y_k}{x_i}\lor \bigwedge\limits_{e_k= \setdefinition{v_i, v_j}}  \depa{y_k}{x_i}\lor \bigwedge\limits_{e_k= \setdefinition{v_i, v_j}}  \depa{y_k}{x_i}$$
The choice of $x_i$ or $x_j$ to appear in the formula is irrelevant. 
The idea is that if there is an edge $e_k$ between two nodes $v_i, v_j$ and accordingly $s_i(y_k)=1 = s_j(y_k)$ then the two assignments cannot be in the same split of team.
This is always true because in that case the assignments $s_i, s_j$ cannot agree on any of $x_i$ or $x_j$. 
Since, by (3.), we have $s_i(x_i)=0$ but there is an edge to $v_j$ and we have $s_j(x_i)=1$.
The desired result is achieved by the following claim.
\begin{claim}
	$\calG$ is $3$-colourable iff $\{s_1,\dots,s_n\}\models \Phi$
\end{claim}

\begin{claimproof}
	``$\Rightarrow$'': Let $V_1,V_2, V_3$ be the distribution of $V$ into three colours. 
    Consequently, for every $v \in  V $ we have an $r\leq 3$ such that $v\in V_r$. 
    Moreover, for every $v_i, v_j\in V_r$ there is no $\ell$ s.t.\ $e_\ell=\setdefinition{v_i, v_j}$. 
    Let $T_r =\setdefinition{s_i \mid v_i\in V_r}$ for each $r\leq 3$, then we show that $\bigcup\limits_ {r<3} T_r = T$ and $T_r\models \phi$. 
    This will prove that $T\models \Phi$ because we can split $T$ into three sub-teams such that each satisfies the disjunct.

    Since for each $s_i, s_j \in T_r$, there is no edge $e_\ell = \setdefinition{{v_i, v_j}}$ this implies that for the tuple $y_\ell$, we have $s_i(y_\ell) \neq s_j(y_\ell)$ and thereby making the dependence atom trivially true. 
    Moreover, by $2$-coherency it is enough to check only for pairs $s_i,s_j$ and since the condition holds for every edge, we have $T_r \models \phi$. 
    Since our assumption is that $V$ can be split into three such sets, we have the split of $T$ into three sub-teams. 
    This gives $T\models \Phi$.
    
``$\Leftarrow$'': Conversely, assume that $T$ can be split into three sub-teams each satisfying $\phi$. 
    Then we show that $V_1,V_2, V_3$ is the partition of $V$ into three colours. 
    Let $V_r =\setdefinition{{v_i \mid s_i \in T_r}}$ then $\bigcup\limits_{r\leq 3} = V$ and for any $v_i, v_j \in V_r$ there is no edge between $v_i, v_j$. 
    Suppose to the contrary that there is an edge $e_\ell = \setdefinition{v_i, v_j} $. 
    Then we must have $s_i, s_j \in T_r$ such that $s_i(y_{\ell})= 1 = s_j(y_{\ell})$. 
    That is, $s_i(y_{\ell,1})= 1 = s_j(y_{\ell,1}), \ldots, s_i(y_{\ell, n})= 1 = s_j(y_{\ell, n})$. 
    Since we have that $s_i(x_i)= 0$ whereas $s_j(x_i)= 1$, this implies $\setdefinition{s_i, s_j} \not \models \phi$ which is a contradiction. \qedclaim
\end{claimproof}
This concludes the full proof.
\end{proof}

\begin{restatetheorem}[mc:remaining-cases]
\begin{theorem}
	$\MC$ for the parameters \formula, \variables, or \formuladepth is $\FPT$.
\end{theorem} 
	
\end{restatetheorem}
\begin{proof}
This follows due to Proposition~\ref{parameter-bound} in conjunction with Lemma~\ref{para:relations} and results in earlier section on $\MC$.
\end{proof}
\begin{restatetheorem}[sat:splits]
\begin{theorem}
	$\SAT$ parameterised by \splits is $\FPT$. 
	Moreover, there is an algorithm that solves the problem in $O(2^{\splits(\Phi)}\cdot |\Phi|^{O(1)})$ for any $\Phi\in\PDL$.
\end{theorem}	
\end{restatetheorem}

\begin{proof} 
We consider partial mappings of the kind $t\colon\VAR\to\setdefinition{0,1,c}$.
Intuitively, these mappings are used to find a satisfying assignment in the process of the presented algorithm.

If $t,t'$ are two (partial) mappings then 
$t\oconflict t'$ is the assignment such that 
$$
(t\oconflict t')(x)\dfn
\begin{cases}
	c &, \text{if both are defined and }t(x)\neq t'(x)\\
	t(x) &, \text{if $t(x)$ is defined}\\
	t'(x) &, \text{if $t'(x)$ is defined}\\
	\text{undefined} &,\text{otherwise}.
\end{cases}
$$
We prove the following claim.
\begin{claim}
The formula $\Phi$ is satisfiable if and only if Algorithm \ref{algo-splits} returns a consistent (partial) assignment $s$ that can be extended to a satisfying assignment for $\Phi$ over $\VAR(\Phi)$.
\end{claim}
\begin{claimproof}
We prove using induction on the structure of $\Phi$.

\begin{description}
	\item[Base case] Start with $\Phi=X$ is a variable.
    Then $\Phi$ is satisfiable and ${s}\models \Phi$ such that $s(x)=1$.
    Moreover, such an assignment is returned by the procedure as depicted by line~3 of the algorithm.
    Similarly, the case $\Phi=\neg X$ follows by line~4. 
    The case $\Phi = \depa{P}{Q}$ or $\Phi = \top$ is a special case of a \pdl-formula since this is true under any assignment. 
    Line~6 in our procedure returns such an assignment that can be extended to any consistent assignment.
    Finally, for $\Phi = \bot$, the assignment contains a conflict and can not be extended to a consistent assignment as the algorithm returns ``$\Phi$ is not satisfiable''.
    \item[Induction Step] Notice first that if either of the two operands is $\top$ then this is a special case and triggers lines 9--11 of the algorithm thereby giving the satisfying assignment.

    Suppose now that $\Phi=\psi_1\land \psi_2$ and that the claim is true for $\psi_1$ and $\psi_2$. 
    As a result, both $\psi_1$ and $\psi_2$ are satisfiable if and only if the algorithm returns a satisfying assignment for each.
    Let $S_i$ for $i=0,1$ be such that every consistent $t'_i\in S_i$ can be extended to a satisfying assignment for $\psi_i$. 
    Our claim is that $S_\Phi$ returned by the procedure (line 13) is non-empty and contains a consistent assignment for $\Phi$ if and only if $\Phi$ is satisfiable.
    First note that, by induction hypothesis, $S_i$ contains all the possible partial assignments that satisfy $\psi_i$ for $i=0,1$. 
    Consequently, $S_\Phi$ contains all the possible $\oconflict$-joins of such assignments that can satisfy $\Phi$. 
    Let $\psi_0$ be satisfied by an assignment $s_0$ and $\psi_1$ be satisfied by an assignment $s_1$. 
    Moreover, let $s' \in S_\Phi$ be an assignment such that $s'=s_0 \oconflict s_1$. 
    If $s'$ is consistent then $s'$ can be extended to a satisfying assignment $s$ for $\Phi$ since ${s'}\models \psi_i$ for $i=0,1$. 
    On the other hand if every $s=s_0 \oconflict s_1$ is conflicting then there is no assignment over $\VAR(\Phi) = \VAR(\psi_1) \cup \VAR(\psi_2)$ that satisfies $\Phi$.
    Accordingly, $\Phi$ is not satisfiable.

    The case for split-junction is simpler. 
    Suppose that $\Phi=\psi_0\lor \psi_1$ and that the claim is true for $\psi_1$ and $\psi_2$. 
    Then $\Phi$ is satisfiable if and only if either $\psi_0$ or $\psi_1$ is satisfiable. 
    Since the label $S_\Phi$ for $\Phi$ is the union of all the labels from $\psi_0$ and $\psi_1$, it is enough to check that either the label of $\psi_0$ ($S_0$) or the label of $\psi_1$ ($S_1$) contains a consistent partial assignment.
    By induction hypothesis, this is equivalent to checking whether $\psi_0$ or $\psi_1$ is satisfiable. 
    This completes the case for split-junction and the proof to our claim. \qedclaim
\end{description}
\end{claimproof}
\begin{algorithm}[t]\LinesNumbered
\caption{SAT-algorithm for \splits which tries to find a satisfying singleton team.}\label{algo-splits}
\SetKwInOut{Input}{Input}
\SetKwInOut{Output}{Output}
\Input{$\PDL$-formula $\Phi$ represented by a syntax tree}
\Output{A set of partial assignments $S$ such that $S\models \Phi$ or ``$\Phi$ is not satisfiable''}
\Begin {
	\ForEach{Leaf $\ell$ of the syntax tree}{
		\lIf{$\ell = X$ is a variable}{$S_\ell \leftarrow \setdefinition{\setdefinition{x\mapsto1}}$}
		\lElseIf{$\ell = \neg X $ is a negated variable}{$S_\ell \leftarrow \setdefinition{\setdefinition{x\mapsto 0}}$}
		\lElseIf{$\ell=\bot$}{pick any $x\in\VAR$ and set $S_\ell \leftarrow \setdefinition{\setdefinition{x\mapsto c}}$}
		\lElse(\hfill\small\ttfamily // case $\top$ or $\depa{\cdot}{\cdot}$) {$S_\ell\leftarrow\setdefinition{1}$}
	}
	\ForEach{Inner node $\ell$ of the syntax tree in bottom-up order}{
		Let $\ell_{0}, \ell_{1}$ be the children of $\ell$ with $S_0, S_1$ the resp.~sets partial assignments\;
		\uIf{$1\in S_i$}{
			\lIf{$\ell$ is a conjunction}{$S_\ell\leftarrow S_{1-i}$}
			\lElse(\hfill\small\ttfamily // empty split for split-junction){$S_\ell\leftarrow\setdefinition{1}$}
		}
		\uElseIf{$\ell$ is a conjunction}{
			\lForEach{$s_0\in S_0$ and $s_1\in S_1$}{
				$S_\ell\leftarrow S_\ell\cup\setdefinition{{s_0\oconflict s_1}}$
			}
		}
		\Else{
			\lForEach{$s_0\in S_0$ and $s_1\in S_1$}{
				$S_\ell\leftarrow S_\ell\cup\setdefinition{s_0, s_1}$
			}
		}
	}
	\lIf{there exists a non-conflicting assignment $s$ in root node}{
		\Return $s$}
	\lElse{\Return ``$\Phi$ is not satisfiable''}
}

\end{algorithm}
Finally, notice that the label size adds at the occurrence of a split-junction. 
That is, we keep all the assignment candidates separate and each such candidate is present in the label for split-junction node. 
In contrast, at conjunction nodes, we `join' the assignments and for this reason, the label size is the product of the two labels.
Notice that we do not get a blow-up in the number of conjunction.
This is because, initially the label size for each node is $1$ and only at a split junction, the size increases.
This implies that the maximum size for any label is bounded by  $2^{\splits}$.
As a consequence, the above algorithm runs in polynomial time in the input and exponential in the parameter. 
\end{proof}
Figure~\ref{ex1:tree-splits} presents an example of using the above algorithm.
To simplify the notation, we consider the assignment labels of the form $\{x_i,\neg x_j\}$ rather than $\{x_i\mapsto 1, x_j\mapsto 0\}$.

\begin{figure}
\centering
\begin{tikzpicture}[scale=0.85,
	level 1/.style={sibling distance=7em, level distance= 3 em},
	level 2/.style={sibling distance=5.5em, level distance= 4 em}, 
	level 3/.style={sibling distance=2.7em, level distance= 4 em}, 
	level 4/.style={sibling distance=2.7em, level distance= 4 em},  
	every node/.style={scale=0.8},
	edge from parent/.style={thin,-,black, draw},
	 leaf/.style = {draw, circle}]
\node  {\textcolor{black}{$\land$}}
	child {node {$\land$} 
		child {node {$\land$}
			child {node {$x_4$}}
			child {node {$x_1$}}}
		child {node {$\neg x_2$}}}
	child {node {$\land$} 
		child {node {$\lor$}
			child {node {$\land$}
				child{node {$x_1$}}
				child{node {$x_2$}}}
			child {node {$\scriptstyle \depa{x_3}{x_4}$}}}		
		child {node {$\lor$}	
				child{node {$x_3$}} 
				child{node {$\neg x_1$}}}};
\end{tikzpicture} 
\hfill
\begin{tikzpicture}[scale=0.85,
	level 1/.style={sibling distance=10em, level distance= 3 em},
	level 2/.style={sibling distance=6em, level distance= 4 em}, 
	level 3/.style={sibling distance=3.5em, level distance= 4 em}, 
	level 4/.style={sibling distance=4em, level distance= 4 em},  
	every node/.style={rectangle, draw, scale=0.8},
	edge from parent/.style={thin,-,black, draw},
	 leaf/.style = {draw, circle}]
\node  {\textcolor{black}{$\setdefinition{x_4,x_1,\neg x_2,x_3}, \setdefinition{x_4,\neg x_2, x_1^c}$}}
	child {node {$\setdefinition{x_4,x_1,\neg x_2}$} 
		child {node {$\setdefinition{x_4,x_1}$}
			child {node {$\setdefinition{x_4}$}}
			child {node {$\setdefinition{x_1}$}}}
		child {node {$\setdefinition{\neg x_2} $}}}
	child {node {$\setdefinition{x_3}, \setdefinition{\neg x_1}$} 
		child {node {$1$}
			child {node {$\setdefinition{x_1, x_2}$}
				child{node {$\setdefinition{x_1}$}}
				child{node {$\setdefinition{ x_2}$}}}
			child {node {$1$}}}			
			child {node {$\setdefinition{x_3}, \setdefinition{\neg x_1}$}
				child{node {$\setdefinition{x_3}$}} 
				child{node {$\setdefinition{\neg x_1}$}}}};
\end{tikzpicture}\caption{(left) syntax tree of example formula, and (right) computation of Algorithm~\ref{algo-splits}. Notation: $x$/$\lnot x$/$x^c$ means a variable is set to true/false/conflict. Clearly, $\setdefinition{x_4,x_1,\neg x_2,x_3}$ satisfies the formula.}\label{ex1:tree-splits}
\end{figure}
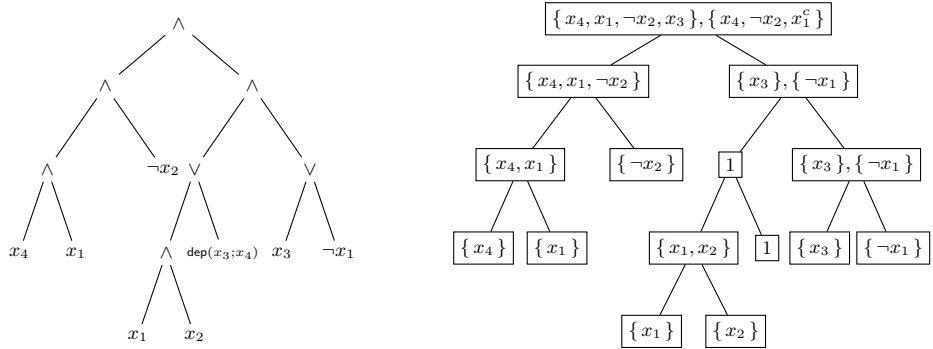

\begin{restatetheorem}[sat:remaining-cases]
\begin{theorem}
	$\SAT$ parameterised by \variables, \formula, or \formuladepth is in $\FPT$.
\end{theorem} 
\end{restatetheorem}
\begin{proof}
Regarding $\variables$, the question \pdl-$\SAT$ boils down to \pl-$\SAT$ of finding an assignment for a given propositional logic formula. 
The latter problem, when parameterised by the number of variables in the input formula, is $\FPT$ which implies that the former problem is also $\FPT$.

Note that $\formula=|\Phi|$ and any PP $\Pi$ is $\FPT$ for the parametrisation input-length. 
Consequently, $\SAT$ parameterised by $\formula$ is $\FPT$. 

If a formula $\Phi$ has $\formuladepth=d$ then there are $\leq 2^d$ leaves and $\leq 2^d$ internal nodes accordingly we have $\formula \leq 2^{2d}$ which shows membership $\FPT$ parameterised by \formuladepth.
\end{proof}

\begin{restatetheorem}[m-sat:all-cases]
\begin{theorem}
	$\mSAT$ parameterised by $\variables$, $\formula$, or \linebreak $\formuladepth$ is in $\FPT$.
\end{theorem}	
\end{restatetheorem}

\begin{proof}
In the case of \variables, the maximal team size is $2^{\variables}$. 
 Moreover, there are $2^{2^{\variables}} $ such teams and we can find all the satisfying teams in fpt-time with respect to $\variables$ and then check whether any such team has size $m$.

 For $\formula$ note that $\formula\leq|x|$ and $\Pi \in \FPT$ for any problem PP $\Pi$ with parameter input-size. 
 As a result, we have that $\mSAT$ parameterised by $\formula$ is $\FPT$

 For $\formuladepth$ notice that $\formula \leq 2^{2\cdot\formuladepth}$ and thereby the problem is $\FPT$ under this parametrisation.
\end{proof}

%% file: arxiv.bbl
\begin{thebibliography}{10}

\bibitem{DBLP:books/daglib/0037838}
Samson Abramsky, Juha Kontinen, Jouko V{\"{a}}{\"{a}}n{\"{a}}nen, and Heribert
  Vollmer, editors.
\newblock {\em Dependence Logic, Theory and Applications}.
\newblock Springer, 2016.
\newblock URL: \url{https://doi.org/10.1007/978-3-319-31803-5}, \href
  {http://dx.doi.org/10.1007/978-3-319-31803-5}
  {\path{doi:10.1007/978-3-319-31803-5}}.

\bibitem{blsius_et_al:LIPIcs:2017:6920}
Thomas Bl{\"a}sius, Tobias Friedrich, and Martin Schirneck.
\newblock {The Parameterized Complexity of Dependency Detection in Relational
  Databases}.
\newblock In Jiong Guo and Danny Hermelin, editors, {\em 11th International
  Symposium on Parameterized and Exact Computation (IPEC 2016)}, volume~63 of
  {\em Leibniz International Proceedings in Informatics (LIPIcs)}, pages
  6:1--6:13, Dagstuhl, Germany, 2017. Schloss Dagstuhl--Leibniz-Zentrum fuer
  Informatik.
\newblock URL: \url{http://drops.dagstuhl.de/opus/volltexte/2017/6920}, \href
  {http://dx.doi.org/10.4230/LIPIcs.IPEC.2016.6}
  {\path{doi:10.4230/LIPIcs.IPEC.2016.6}}.

\bibitem{DBLP:conf/icdt/ChenM15}
Hubie Chen and Stefan Mengel.
\newblock A trichotomy in the complexity of counting answers to conjunctive
  queries.
\newblock In Marcelo Arenas and Mart{\'{\i}}n Ugarte, editors, {\em 18th
  International Conference on Database Theory, {ICDT} 2015, March 23-27, 2015,
  Brussels, Belgium}, volume~31 of {\em LIPIcs}, pages 110--126. Schloss
  Dagstuhl - Leibniz-Zentrum fuer Informatik, 2015.
\newblock URL: \url{https://doi.org/10.4230/LIPIcs.ICDT.2015.110}, \href
  {http://dx.doi.org/10.4230/LIPIcs.ICDT.2015.110}
  {\path{doi:10.4230/LIPIcs.ICDT.2015.110}}.

\bibitem{DBLP:conf/pods/ChenM16}
Hubie Chen and Stefan Mengel.
\newblock Counting answers to existential positive queries: {A} complexity
  classification.
\newblock In Tova Milo and Wang{-}Chiew Tan, editors, {\em Proceedings of the
  35th {ACM} {SIGMOD-SIGACT-SIGAI} Symposium on Principles of Database Systems,
  {PODS} 2016, San Francisco, CA, USA, June 26 - July 01, 2016}, pages
  315--326. {ACM}, 2016.
\newblock URL: \url{https://doi.org/10.1145/2902251.2902279}, \href
  {http://dx.doi.org/10.1145/2902251.2902279}
  {\path{doi:10.1145/2902251.2902279}}.

\bibitem{DBLP:conf/lics/ChenM17}
Hubie Chen and Stefan Mengel.
\newblock The logic of counting query answers.
\newblock In {\em 32nd Annual {ACM/IEEE} Symposium on Logic in Computer
  Science, {LICS} 2017, Reykjavik, Iceland, June 20-23, 2017}, pages 1--12.
  {IEEE} Computer Society, 2017.
\newblock URL: \url{https://doi.org/10.1109/LICS.2017.8005085}, \href
  {http://dx.doi.org/10.1109/LICS.2017.8005085}
  {\path{doi:10.1109/LICS.2017.8005085}}.

\bibitem{DBLP:journals/toct/ChenM15}
Hubie Chen and Moritz M{\"{u}}ller.
\newblock The fine classification of conjunctive queries and parameterized
  logarithmic space.
\newblock {\em {TOCT}}, 7(2):7:1--7:27, 2015.
\newblock URL: \url{https://doi.org/10.1145/2751316}, \href
  {http://dx.doi.org/10.1145/2751316} {\path{doi:10.1145/2751316}}.

\bibitem{DBLP:conf/tbillc/CiardelliGR11}
Ivano Ciardelli, Jeroen Groenendijk, and Floris Roelofsen.
\newblock Towards a logic of information exchange - an inquisitive witness
  semantics.
\newblock In {\em Logic, Language, and Computation - 9th International Tbilisi
  Symposium on Logic, Language, and Computation, TbiLLC 2011, Kutaisi, Georgia,
  September 26-30, 2011, Revised Selected Papers}, pages 51--72, 2011.
\newblock URL: \url{https://doi.org/10.1007/978-3-642-36976-6\_6}, \href
  {http://dx.doi.org/10.1007/978-3-642-36976-6\_6}
  {\path{doi:10.1007/978-3-642-36976-6\_6}}.

\bibitem{DBLP:conf/stoc/Cook71}
Stephen~A. Cook.
\newblock The complexity of theorem-proving procedures.
\newblock In Michael~A. Harrison, Ranan~B. Banerji, and Jeffrey~D. Ullman,
  editors, {\em Proceedings of the 3rd Annual {ACM} Symposium on Theory of
  Computing, May 3-5, 1971, Shaker Heights, Ohio, {USA}}, pages 151--158.
  {ACM}, 1971.
\newblock URL: \url{https://doi.org/10.1145/800157.805047}, \href
  {http://dx.doi.org/10.1145/800157.805047} {\path{doi:10.1145/800157.805047}}.

\bibitem{DBLP:books/sp/CyganFKLMPPS15}
Marek Cygan, Fedor~V. Fomin, Lukasz Kowalik, Daniel Lokshtanov, D{\'{a}}niel
  Marx, Marcin Pilipczuk, Michal Pilipczuk, and Saket Saurabh.
\newblock {\em Parameterized Algorithms}.
\newblock Springer, 2015.
\newblock URL: \url{https://doi.org/10.1007/978-3-319-21275-3}, \href
  {http://dx.doi.org/10.1007/978-3-319-21275-3}
  {\path{doi:10.1007/978-3-319-21275-3}}.

\bibitem{DBLP:series/txcs/DowneyF13}
Rodney~G. Downey and Michael~R. Fellows.
\newblock {\em Fundamentals of Parameterized Complexity}.
\newblock Texts in Computer Science. Springer, 2013.
\newblock URL: \url{https://doi.org/10.1007/978-1-4471-5559-1}, \href
  {http://dx.doi.org/10.1007/978-1-4471-5559-1}
  {\path{doi:10.1007/978-1-4471-5559-1}}.

\bibitem{DBLP:journals/amai/DurandHKMV18}
Arnaud Durand, Miika Hannula, Juha Kontinen, Arne Meier, and Jonni Virtema.
\newblock Approximation and dependence via multiteam semantics.
\newblock {\em Ann. Math. Artif. Intell.}, 83(3-4):297--320, 2018.
\newblock URL: \url{https://doi.org/10.1007/s10472-017-9568-4}, \href
  {http://dx.doi.org/10.1007/s10472-017-9568-4}
  {\path{doi:10.1007/s10472-017-9568-4}}.

\bibitem{DBLP:conf/foiks/0001HKMV18}
Arnaud Durand, Miika Hannula, Juha Kontinen, Arne Meier, and Jonni Virtema.
\newblock Probabilistic team semantics.
\newblock In Flavio Ferrarotti and Stefan Woltran, editors, {\em Foundations of
  Information and Knowledge Systems - 10th International Symposium, FoIKS 2018,
  Budapest, Hungary, May 14-18, 2018, Proceedings}, volume 10833 of {\em
  Lecture Notes in Computer Science}, pages 186--206. Springer, 2018.
\newblock URL: \url{https://doi.org/10.1007/978-3-319-90050-6\_11}, \href
  {http://dx.doi.org/10.1007/978-3-319-90050-6\_11}
  {\path{doi:10.1007/978-3-319-90050-6\_11}}.

\bibitem{DBLP:journals/mst/0001M15}
Arnaud Durand and Stefan Mengel.
\newblock Structural tractability of counting of solutions to conjunctive
  queries.
\newblock {\em Theory Comput. Syst.}, 57(4):1202--1249, 2015.
\newblock URL: \url{https://doi.org/10.1007/s00224-014-9543-y}, \href
  {http://dx.doi.org/10.1007/s00224-014-9543-y}
  {\path{doi:10.1007/s00224-014-9543-y}}.

\bibitem{DBLP:conf/sofsem/EbbingL12}
Johannes Ebbing and Peter Lohmann.
\newblock Complexity of model checking for modal dependence logic.
\newblock In M{\'{a}}ria Bielikov{\'{a}}, Gerhard Friedrich, Georg Gottlob,
  Stefan Katzenbeisser, and Gy{\"{o}}rgy Tur{\'{a}}n, editors, {\em {SOFSEM}
  2012: Theory and Practice of Computer Science - 38th Conference on Current
  Trends in Theory and Practice of Computer Science, {\v{S}}pindler{\r{u}}v
  Ml{\'{y}}n, Czech Republic, January 21-27, 2012. Proceedings}, volume 7147 of
  {\em Lecture Notes in Computer Science}, pages 226--237. Springer, 2012.
\newblock URL: \url{https://doi.org/10.1007/978-3-642-27660-6\_19}, \href
  {http://dx.doi.org/10.1007/978-3-642-27660-6\_19}
  {\path{doi:10.1007/978-3-642-27660-6\_19}}.

\bibitem{DBLP:series/txtcs/FlumG06}
J{\"{o}}rg Flum and Martin Grohe.
\newblock {\em Parameterized Complexity Theory}.
\newblock Texts in Theoretical Computer Science. An {EATCS} Series. Springer,
  2006.
\newblock URL: \url{https://doi.org/10.1007/3-540-29953-X}, \href
  {http://dx.doi.org/10.1007/3-540-29953-X} {\path{doi:10.1007/3-540-29953-X}}.

\bibitem{DBLP:journals/apal/Galliani12}
Pietro Galliani.
\newblock Inclusion and exclusion dependencies in team semantics - on some
  logics of imperfect information.
\newblock {\em Ann. Pure Appl. Logic}, 163(1):68--84, 2012.
\newblock URL: \url{https://doi.org/10.1016/j.apal.2011.08.005}, \href
  {http://dx.doi.org/10.1016/j.apal.2011.08.005}
  {\path{doi:10.1016/j.apal.2011.08.005}}.

\bibitem{DBLP:journals/tcs/Gradel13}
Erich Gr{\"{a}}del.
\newblock Model-checking games for logics of imperfect information.
\newblock {\em Theor. Comput. Sci.}, 493:2--14, 2013.
\newblock URL: \url{https://doi.org/10.1016/j.tcs.2012.10.033}, \href
  {http://dx.doi.org/10.1016/j.tcs.2012.10.033}
  {\path{doi:10.1016/j.tcs.2012.10.033}}.

\bibitem{grdel_et_al:DR:2016:5508}
Erich Gr{\"a}del, Juha Kontinen, Jouka V{\"a}{\"a}n{\"a}nen, and Heribert
  Vollmer.
\newblock {Logics for Dependence and Independence (Dagstuhl Seminar 15261)}.
\newblock {\em Dagstuhl Reports}, 5(6):70--85, 2016.
\newblock URL: \url{http://drops.dagstuhl.de/opus/volltexte/2016/5508}, \href
  {http://dx.doi.org/10.4230/DagRep.5.6.70} {\path{doi:10.4230/DagRep.5.6.70}}.

\bibitem{DBLP:conf/stoc/GroheSS01}
Martin Grohe, Thomas Schwentick, and Luc Segoufin.
\newblock When is the evaluation of conjunctive queries tractable?
\newblock In Jeffrey~Scott Vitter, Paul~G. Spirakis, and Mihalis Yannakakis,
  editors, {\em Proceedings on 33rd Annual {ACM} Symposium on Theory of
  Computing, July 6-8, 2001, Heraklion, Crete, Greece}, pages 657--666. {ACM},
  2001.
\newblock URL: \url{https://doi.org/10.1145/380752.380867}, \href
  {http://dx.doi.org/10.1145/380752.380867} {\path{doi:10.1145/380752.380867}}.

\bibitem{DBLP:conf/lfcs/HannulaKV18}
Miika Hannula, Juha Kontinen, and Jonni Virtema.
\newblock Polyteam semantics.
\newblock In Sergei~N. Art{\"{e}}mov and Anil Nerode, editors, {\em Logical
  Foundations of Computer Science - International Symposium, {LFCS} 2018,
  Deerfield Beach, FL, USA, January 8-11, 2018, Proceedings}, volume 10703 of
  {\em Lecture Notes in Computer Science}, pages 190--210. Springer, 2018.
\newblock URL: \url{https://doi.org/10.1007/978-3-319-72056-2\_12}, \href
  {http://dx.doi.org/10.1007/978-3-319-72056-2\_12}
  {\path{doi:10.1007/978-3-319-72056-2\_12}}.

\bibitem{DBLP:journals/tocl/HannulaKVV18}
Miika Hannula, Juha Kontinen, Jonni Virtema, and Heribert Vollmer.
\newblock Complexity of propositional logics in team semantic.
\newblock {\em {ACM} Trans. Comput. Log.}, 19(1):2:1--2:14, 2018.
\newblock URL: \url{https://doi.org/10.1145/3157054}, \href
  {http://dx.doi.org/10.1145/3157054} {\path{doi:10.1145/3157054}}.

\bibitem{hkmv19}
Lauri Hella, Antti Kuusisto, Arne Meier, and Jonni Virtema.
\newblock Model checking and validity in propositional and modal inclusion
  logics.
\newblock {\em J. Log. Comput.}, 2019.
\newblock \href {http://dx.doi.org/10.1093/logcom/exz008}
  {\path{doi:10.1093/logcom/exz008}}.

\bibitem{DBLP:journals/logcom/KontinenMSV17}
Juha Kontinen, Julian{-}Steffen M{\"{u}}ller, Henning Schnoor, and Heribert
  Vollmer.
\newblock Modal independence logic.
\newblock {\em J. Log. Comput.}, 27(5):1333--1352, 2017.
\newblock URL: \url{https://doi.org/10.1093/logcom/exw019}, \href
  {http://dx.doi.org/10.1093/logcom/exw019} {\path{doi:10.1093/logcom/exw019}}.

\bibitem{DBLP:conf/time/KrebsMV15}
Andreas Krebs, Arne Meier, and Jonni Virtema.
\newblock A team based variant of {CTL}.
\newblock In Fabio Grandi, Martin Lange, and Alessio Lomuscio, editors, {\em
  22nd International Symposium on Temporal Representation and Reasoning, {TIME}
  2015, Kassel, Germany, September 23-25, 2015}, pages 140--149. {IEEE}
  Computer Society, 2015.
\newblock URL: \url{https://doi.org/10.1109/TIME.2015.11}, \href
  {http://dx.doi.org/10.1109/TIME.2015.11} {\path{doi:10.1109/TIME.2015.11}}.

\bibitem{DBLP:conf/mfcs/KrebsMV018}
Andreas Krebs, Arne Meier, Jonni Virtema, and Martin Zimmermann.
\newblock Team semantics for the specification and verification of
  hyperproperties.
\newblock In Igor Potapov, Paul~G. Spirakis, and James Worrell, editors, {\em
  43rd International Symposium on Mathematical Foundations of Computer Science,
  {MFCS} 2018, August 27-31, 2018, Liverpool, {UK}}, volume 117 of {\em
  LIPIcs}, pages 10:1--10:16. Schloss Dagstuhl - Leibniz-Zentrum fuer
  Informatik, 2018.
\newblock URL: \url{https://doi.org/10.4230/LIPIcs.MFCS.2018.10}, \href
  {http://dx.doi.org/10.4230/LIPIcs.MFCS.2018.10}
  {\path{doi:10.4230/LIPIcs.MFCS.2018.10}}.

\bibitem{levin73}
Leonid Levin.
\newblock Universal search problems.
\newblock {\em Problems of Information Transmission}, 9(3):115--116, 1973.

\bibitem{DBLP:journals/sLogica/LohmannV13}
Peter Lohmann and Heribert Vollmer.
\newblock Complexity results for modal dependence logic.
\newblock {\em Studia Logica}, 101(2):343--366, 2013.
\newblock URL: \url{https://doi.org/10.1007/s11225-013-9483-6}, \href
  {http://dx.doi.org/10.1007/s11225-013-9483-6}
  {\path{doi:10.1007/s11225-013-9483-6}}.

\bibitem{DBLP:journals/tocl/LuckMS17}
Martin L{\"{u}}ck, Arne Meier, and Irena Schindler.
\newblock Parameterised complexity of satisfiability in temporal logic.
\newblock {\em {ACM} Trans. Comput. Log.}, 18(1):1:1--1:32, 2017.
\newblock URL: \url{https://doi.org/10.1145/3001835}, \href
  {http://dx.doi.org/10.1145/3001835} {\path{doi:10.1145/3001835}}.

\bibitem{DBLP:conf/foiks/MeierR18}
Arne Meier and Christian Reinbold.
\newblock Enumeration complexity of poor man's propositional dependence logic.
\newblock In Flavio Ferrarotti and Stefan Woltran, editors, {\em Foundations of
  Information and Knowledge Systems - 10th International Symposium, FoIKS 2018,
  Budapest, Hungary, May 14-18, 2018, Proceedings}, volume 10833 of {\em
  Lecture Notes in Computer Science}, pages 303--321. Springer, 2018.
\newblock URL: \url{https://doi.org/10.1007/978-3-319-90050-6\_17}, \href
  {http://dx.doi.org/10.1007/978-3-319-90050-6\_17}
  {\path{doi:10.1007/978-3-319-90050-6\_17}}.

\bibitem{DBLP:books/daglib/0092426}
Nicholas Pippenger.
\newblock {\em Theories of computability}.
\newblock Cambridge University Press, 1997.

\bibitem{DBLP:journals/jct/RobertsonS86}
Neil Robertson and Paul~D. Seymour.
\newblock Graph minors. v. excluding a planar graph.
\newblock {\em J. Comb. Theory, Ser. {B}}, 41(1):92--114, 1986.
\newblock URL: \url{https://doi.org/10.1016/0095-8956(86)90030-4}, \href
  {http://dx.doi.org/10.1016/0095-8956(86)90030-4}
  {\path{doi:10.1016/0095-8956(86)90030-4}}.

\bibitem{DBLP:series/faia/SamerS09}
Marko Samer and Stefan Szeider.
\newblock Fixed-parameter tractability.
\newblock In Armin Biere, Marijn Heule, Hans van Maaren, and Toby Walsh,
  editors, {\em Handbook of Satisfiability}, volume 185 of {\em Frontiers in
  Artificial Intelligence and Applications}, pages 425--454. {IOS} Press, 2009.
\newblock URL: \url{https://doi.org/10.3233/978-1-58603-929-5-425}, \href
  {http://dx.doi.org/10.3233/978-1-58603-929-5-425}
  {\path{doi:10.3233/978-1-58603-929-5-425}}.

\bibitem{DBLP:journals/jda/SamerS10}
Marko Samer and Stefan Szeider.
\newblock Algorithms for propositional model counting.
\newblock {\em J. Discrete Algorithms}, 8(1):50--64, 2010.

\bibitem{DBLP:books/daglib/0030191}
Jouko~A. V{\"{a}}{\"{a}}n{\"{a}}nen.
\newblock {\em Dependence Logic - {A} New Approach to Independence Friendly
  Logic}, volume~70 of {\em London Mathematical Society student texts}.
\newblock Cambridge University Press, 2007.
\newblock URL:
  \url{http://www.cambridge.org/de/knowledge/isbn/item1164246/?site\_locale=de\_DE}.

\bibitem{v08}
Jouko~A. V{\"{a}}{\"{a}}n{\"{a}}nen.
\newblock Modal dependence logic.
\newblock In Krzystof Apt and Robert van Rooij, editors, {\em New Perspectives
  on Games and Interaction}. Amsterdam University Press, 2008.

\bibitem{DBLP:journals/apal/YangV16}
Fan Yang and Jouko V{\"{a}}{\"{a}}n{\"{a}}nen.
\newblock Propositional logics of dependence.
\newblock {\em Ann. Pure Appl. Logic}, 167(7):557--589, 2016.
\newblock URL: \url{https://doi.org/10.1016/j.apal.2016.03.003}, \href
  {http://dx.doi.org/10.1016/j.apal.2016.03.003}
  {\path{doi:10.1016/j.apal.2016.03.003}}.

\end{thebibliography}
